\documentclass[compact]{llncs}
\usepackage{color}
\usepackage{graphicx}
\usepackage{amsfonts}
\usepackage{latexsym,amssymb,amsmath}
\usepackage{todonotes}
\usepackage[ruled]{algorithm}
\usepackage[noend]{algpseudocode}
\algdef{SE}[DOWHILE]{Do}{doWhile}{\algorithmicdo}[1]{\algorithmicwhile\ #1}
\usepackage{comment}
\usepackage{epstopdf}
\usepackage{soul}
\usepackage[colorlinks=true, allcolors=blue]{hyperref}
\usepackage{cleveref}

\begin{document}

\title{Linear Search for Capturing an Oblivious Mobile Target in the Sender/Receiver Model\thanks{A preliminary version of this paper appears in the Proceedings of the International Symposium on Algorithmics of Wireless Networks (ALGOWIN 2025).}}

\author{
Khaled Jawhar\inst{1}\and
Evangelos Kranakis\inst{1}\inst{2}
}

\institute{
School of Computer Science, Carleton University, Ottawa, ON, Canada.
\and
Research supported in part by NSERC Discovery grant.
}
\maketitle
\begin{abstract}
We consider linear search for capturing an oblivious moving target by two autonomous robots with different communicating abilities. Both robots can communicate Face-to-Face (F2F) when co-located but in addition one robot is a Sender (can also send messages wirelessly) and the other also a Receiver (can also receive messages wirelessly). This is known as Sender/Receiver (S/R, for short) communication model. The robots can move with max speed $1$. The moving target starts at distance $d$ from the origin and can move either with speed $v<1$ away from the origin in the ``away'' model or with speed $v \geq 0$ toward the origin in the ``toward'' model. We assume that the direction of motion of the target (i.e., whether it is the away or toward model) is known to the robots in advance. To capture the target the two robots must be co-located with it. 

We design new linear search algorithms and analyze the competitive ratio of the time required to capture the target. The approach takes into account various scenarios related to what the robots know about the search environment (e.g., starting distance or speed of the mobile, away or toward model, or a combination thereof). Our study contributes to understanding how asymmetric communication affects the competitive ratio of linear search. 

\vspace{.1cm}

\vspace{.1cm}
\noindent
{\bf Khaled Jawhar is eligible for the best student paper award.}
\end{abstract}

\section{Introduction}

Linear search and evacuation by one or more autonomous mobile robots (or agents) have been studied extensively and are applicable to many areas of theoretical computer science, including data mining, crawling, and surveillance, thus making them an area of significant interest. Linear search was first proposed for stochastic and game theoretic systems in \cite{beck1964linear,bellman1963optimal}, independently. Deterministic search by a single robot operating on the line was subsequently investigated by several researchers including~\cite{baezayates1993searching,baezayates1995parallel} and in other topologies, like star in \cite{gal1974minimax}. Additional work on linear search can also be found in \cite{alpern03}.

Linear search with multiple robots is also referred to as group search and is an important task arising from the need to design algorithms for multi-agent systems. More recently, group search has attracted the attention of researchers in distributed computing in order to understand the impact of communication faults~\cite{DBLP:journals/ijfcs/CzyzowiczGKKNOS21,PODC16} on linear search. In fact, one wants to know how the knowledge the robots have about the search environment affects the competitive ratio of linear search~\cite{kranakis2024survey}. In particular, there is interest in designing algorithms and analyzing tradeoffs involving time, mobility, and communication model for finding a target. Proposed algorithms employ cooperating, communicating autonomous mobile agents in a distributed setting and operate over continuous domains (typically the infinite line). In such settings, designed algorithms should be fault-tolerant and their performance is measured by the competitive ratio. 


The focus of the present paper is on linear search for capturing a mobile target; it involves two cooperating searchers communicating under the S/R model, whereby both robots can communicate in the F2F model but in addition there is communication asymmetry in that one robot is a Sender (can also send messages wirelessly) and the other also a Receiver (can also receive messages wirelessly). The S/R model was initiated in \cite{czyzowicz2021groupevac} in order to understand the impact of asymmetric communication. A key aspect of our analysis will involve comparing an algorithm designed for a specific knowledge model, where the robots have limited information about the input, to a full knowledge model, where the robots know everything about the input. Such information may include the half-line where the mobile started, its speed, direction of movement, or distance from the origin.




\subsection{Model, Preliminaries, and Notation}

The search domain is the bidirectional infinite line, in which the robots can move in either direction without affecting their speeds. The robots start at the origin and can travel with maximum speed $1$. The target is oblivious, starts at distance $d$ from the origin and can move either away from or toward the origin with a maximum speed $v$. We abbreviate mobile target either as mobile or target. If the target is moving away from the origin, we also assume that $v < 1$; if it is moving toward the origin, we allow $v \geq 0$ to be arbitrary. The robots and the mobile start at the same time. If $v=0$ the target is static, so our approach generalizes linear search case for a static target. We consider three knowledge models, cf.,~\cite{coleman2023line}, for the robots: in the NoDistance, $v$ is known and $d$ unknown, in the NoSpeed, $v$ is unknown and $d$ known, and in the NoKnowledge, neither $v$ nor $d$ is known. 

A search algorithm for two robots is a complete description of their trajectories. For an algorithm $A$ and an input instance $I$ of the problem in the knowledge model, $T_A(I)$ is the time it takes algorithm $A$ to solve instance $I$. If $T_{opt}(I)$ is the optimal time of an offline algorithm for the same instance $I$, then the competitive ratio for the capture time of an online algorithm $A$ is defined by the ratio $CR_A := \sup_{I} T_A(I)/T_{opt}(I).$  If ${\cal A} $ is a class of algorithms solving an online version then its competitive ratio is defined by $CR_{{\cal A}}:= \inf_{A \in {\cal A}} CR_A$; we omit the subscripts $A$ and ${\cal A}$ when they are easily understood from context. When the target is static it is customary to refer to capture as evacuation. Our goal is to design capture algorithms that achieve the best competitive ratio.

The robots can always communicate F2F (only when colocated). Additional communication is possible but is asymmetric in that one of the robots can also send communication wirelessly and is designated as the sender (denoted by $S$), while the other robot is designated as the receiver (denoted by $R$) and can also receive information wirelessly. Both robots are equipped with pedometers and computing capabilities, allowing them to deduce the location of the other robot from relevant communications exchanged. $S$ and $R$ can't switch roles.

\subsection{Related Work}


Linear search in a distributed setting with multiple robots subjected to possible crash faults was initiated in~\cite{DBLP:journals/dc/CzyzowiczKKNO19} and for Byzantine faults in~\cite{DBLP:journals/ijfcs/CzyzowiczGKKNOS21}. One aspect of our current research is related to the communication model being used. The S/R communication model was introduced for an infinite line in~\cite{czyzowicz2021groupevac};
for the case of two robots and an unknown static (non-mobile) target, it was shown that there exists an evacuation algorithm with competitive ratio $3 + 2\sqrt{2}$, which was also proven to be optimal among all possible linear search algorithms. Also related is the work of \cite{algowinJawharK24} which studies bike-assisted evacuation in the S/R model. The interested reader can find additional discussion and review of recent research on other search models in \cite{kranakis2024survey}.

Another aspect of the present research is related to understanding how the mobility of the target affects the competitive ratio of linear search. McCabe \cite{mccabe1974} was first to investigate this problem in a stochastic setting whereby the mobile follows a Bernoulli random walk on the integers. In the deterministic (continuous) setting Alpern and Gal~\cite{alpern03}[p. 134, Equation 8.25] were the first to compute the optimal competitive ratio of search as $1 + \frac{8(1+v)}{(1-v)^2}$ for a target moving with speed $0 \leq v <1$. Extensions of this work, were investigated in~\cite{coleman2023line} where also competitive ratio tradeoffs were analyzed depending on whether the target is moving toward or away from the origin of the inifinte line. For each of these two cases four subcases arise depending on which of the two parameters $d$ and $v$ are known to the searcher. Recent extensions of this work can be found in~\cite{ivanov24} where tight (in fact, optimal) bounds are obtained when the target is moving away from the origin. 






The goal of our present paper is twofold. First we generalize the work of~\cite{czyzowicz2021groupevac} from a static target to a mobile target, and second we extend   the analysis of~\cite{coleman2023line}
from one to two robots in the asymmetric S/R model. We extend the work of~\cite{czyzowicz2021groupevac} by generalizing their model to accommodate a moving target, where their setting becomes a special case of ours when the target speed is zero. With respect to~\cite{coleman2023line}, we introduce new algorithms under an asymmetric Sender/Receiver (S/R) communication model involving two robots. This allows us to achieve improved competitive ratios compared to the single-robot setting analyzed in~\cite{coleman2023line}, particularly in scenarios where the robots lack information about the target’s distance, speed, or direction of motion. To this end we give new algorithms and study competitive ratio tradeoffs for linear search in several search models which take into account knowledge of aspects of the target's mobility by two autonomous robots.  To the best of our knowledge this has not been investigated before.

\subsection{Outline and Results of the Paper}


The NoDistance model ($v$ known, $d$ unknown) is studied in Section~\ref{sec: The NoDistance Model}. The main algorithm in the toward model consists of three subalgorithms with respective competitive ratios as follows:  $\frac{\sqrt{v^2+2v+2}+1}{\sqrt{v^2+2v+2}-1}$ for Algorithm~\ref{NoDistanceTowardOppositeDirectionAlgorithm},  $ 1 + \frac{2(a^5 + a^4)}{a^4 + v a^4 + 2a v +v - 1}$ for Algorithm~\ref{ZigZagTillMeetingAlgorithm}, where $a$ is the positive solution of Equation~\eqref{eq:solution}, and $1+\frac1v$ for the waiting algorithm. In summary, Algorithm~\ref{NoDistanceTowardOppositeDirectionAlgorithm} performs best for $0 \leq v \leq \frac{1}{3}$, and the waiting algorithm outperforms the first one when $\frac{1}{3}\leq v \leq 1$. Figure~\ref{Performance} displays the range of $v$ within which the algorithms are optimal. The main Algorithm \ref{NoDistanceAwayS/RAlgorithm} in the away model has competitive ratio $\frac{\sqrt{v^2-2v+2}+1}{\sqrt{v^2-2v+2}-1}$ (which for $v=0$ is equal to $3+2\sqrt{2}$, thus generalizing the upper bound in~\cite{czyzowicz2021groupevac} for arbitrary $v<1$.)

The NoSpeed model ($v$ unknown, $d$ known) is studied in Section~\ref{sec: The NoSpeed Model}. For the toward model Algorithm~\ref{NoSpeedTowardS/RAlgorithm} has competitive ratio $3$, while for the away model Algorithm \ref{NoSpeedAwayS/RAlgorithm} has competitive ratio $1+O(u^{\frac{10}{3}} \log u )$, where $u = \frac1{1-v}$. (The latter result is rather surprising given the communication asymmetry of the two searchers and should be compared to the single searcher case, where combining~\cite{coleman2023line} with the more recent~\cite{ivanov24} the tight bounds $O( u^{4-(\log_2 \log_2 u)^{-2}})$, if $u>4$ and $\Omega (u^{4-\epsilon})$, for any $\epsilon >0$ were obtained.) We also prove in Theorem~\ref{thm:NoSpeedAwayS/RAlgorithmLB} that $1 + \frac2{1-v}$, is a lower bound for the competitive ratio of any algorithm in the NoSpeed away model.

Finally, the NoKnowledge model (both $d,v$ are unknown) is studied in Section~\ref{sec: The NoKnowledge Model}. In the toward model the optimal competitive ratio is $1+\frac{1}{v}$. In the away model we prove that $1 + O( M^{\frac{16}{3}} \log(M) \log \log^{\frac{3}{2}} M )
$ is an upper bound for the competitive ratio, where $M = \max\{d,\frac{1}{1-v}\}$. 

\section{The NoDistance Model}
\label{sec: The NoDistance Model}

In this section, we consider the NoDistance model, in which $d$ is unknown, but $v$ is known to the robots. We distinguish the cases where the target is moving away from or toward the origin.

\subsection{Target moving toward the origin}
The target moves toward the origin from either direction at any speed $v$. Agents $S$ and $R$ can move in any direction with a maximum speed of 1. Three algorithms are considered, two of which perform best over different ranges of the target's speed $v$.

\subsubsection{Opposite Direction Algorithm.}
Assume that $R$ moves in one direction with a maximum unit speed, and $S$ moves in the other direction with speed $u$. Two cases are considered depending on which robot finds the target first. If $S$ finds the target first, then it informs $R$, which proceeds to the target. Otherwise, if $R$ finds the target first, it switches its direction to catch up to $S$, and then both robots move to the target with unit speed.

\begin{algorithm}[H]
\caption{NoDistanceTowardOppositeDirection
}
\label{NoDistanceTowardOppositeDirectionAlgorithm}
\begin{algorithmic}[1]
\State {$S$ moves left with speed $u$ and
$R$ moves right with speed $1$;}
\If {$R$ finds the target first}
    \State{It changes direction and catches up to $S$;}
    \State{The two robots change direction and move toward the target with unit speed;}
\Else
    \If {$S$ finds the target first} 
        \State{It informs $R$ that the target has been found and moves with the target;}
        \State{$R$ switches direction and moves with unit speed to the target;}
    \EndIf
\EndIf
\end{algorithmic}
\end{algorithm}


\begin{theorem}
\label{NoDistanceTowardOppositeDirectionTheorem}
The competitive ratio of Algorithm~\ref{NoDistanceTowardOppositeDirectionAlgorithm} is 
\begin{equation}
\label{eq:alg1}
\frac{\sqrt{v^2+2v+2}+1}{\sqrt{v^2+2v+2}-1}.
\end{equation}
\end{theorem}

\begin{proof} (Theorem~\ref{NoDistanceTowardOppositeDirectionTheorem})
Assume that $S$ moves with speed $u$ (in the course of the proof we will determine the optimal $u$) and $R$ with unit speed. There are two cases to consider depending on which of $S$ or $R$ finds the target first.
We fix the convention that movement to the right corresponds to the positive x-axis, and movement to the left corresponds to the negative x-axis.
\begin{itemize}
    \item \textbf{$S$ finds the target first.}
    For $S$ to find the target, it needs time $\frac{d}{u + v}$. At that point, the receiver is at a distance of $\frac{d}{u + v} + \frac{du}{u + v}$ from the sender.
    The competitive ratio becomes:
    \begin{align}
        CR &= \label{eq0} \frac{\frac{d}{u+v}}{\frac{d}{1+v}} + \frac{\frac{\frac{d}{u+v}+\frac{du}{u+v}}{1+v}}{\frac{d}{1+v}} = \frac{2+v+u}{u+v}= 1 + \frac{2}{u + v}.
    \end{align}
    \item \textbf{$R$ finds the target first.} The time needed for $R$ to reach the target is
$\frac{d}{1+v}$.
    At this time, $S$ would be away from the origin by a distance $\frac{du}{1+v}$. Thus, for $R$ to capture $S$, it needs time:
    \begin{equation}
        \label{eq2}
        \frac{\frac{d}{1+v}+\frac{du}{1+v}}{1-u} = \frac{d+du}{(1+v)(1-u)}.
    \end{equation}
    Both robots would then be away from the target by:
    \begin{align*}
        \frac{d+du}{1+v} + \frac{du+du^2}{(1+v)(1-u)} - \frac{dv+duv}{(1+v)(1-u)} = \frac{du - dv - duv + d}{(1+v)(1-u)}.
    \end{align*}
    Thus, for both robots to catch the target, they need time:
    \begin{equation}
        \label{eq3}
        \frac{du - dv - duv + d}{(1+v)^2(1-u)}.
    \end{equation}
\end{itemize}

Combining Equations~\eqref{eq2}, and~\eqref{eq3}, the competitive ratio becomes:
\begin{align}
    CR &= \label{eq0a} 1 + \frac{1+u}{1-u} + \frac{u-v-uv+1}{(1+v)(1-u)} = \frac{u-uv+3+v}{(1+v)(1-u)}.
\end{align}

As we will justify below, the optimal value of $u$ is found by solving 
$\frac{u - uv + 3 + v}{(1 + v)(1 - u)} = \frac{2 + v + u}{u + v}$.

We note that both competitive ratio functions derived for the two cases are monotonic in $u$.
Specifically, in the first case where $S$ finds the target first, the competitive ratio $\frac{2+v+u}{u+v}$
  is decreasing in $u$ for $u>0$.
In the second case where $R$ finds the target first, the competitive ratio $\frac{u-uv+3+v}{(1+v)(1-u)}$ is increasing in $u$ for $0<u<1$.
Therefore, equating the two expressions gives the unique value of $u$ that minimizes the maximum competitive ratio between the two cases.
Solving this equation, we obtain
$
2u^2 + 4u + 4uv - 2 = 0,
$
whose solution is
$
u = \sqrt{v^2+2v+2} - v - 1.
$
Substituting $u$ into the competitive ratio formula, we get Equation~\eqref{eq:alg1}.
This proves the theorem.~\qed
\end{proof}

\subsubsection{ZigZag Till Meeting Algorithm.}
Both robots iterate a ZigZag strategy. They start at the origin. During the first iteration, robot $S$ moves in one direction a distance of $x_{0}$, while robot $R$ moves in the opposite direction a distance of $x_{1}$. If neither robot finds the target during this iteration, they reverse their directions and return to meet at some meeting point. If one robot finds the target, it stays with the target. If, upon reversing, one robot does not meet the other at the predetermined meeting point, the robot that arrives at the meeting point will proceed in the direction of the missing robot to catch up with the target.
The algorithm is formalized as follows:
\begin{algorithm}[H]
\caption{ZigZagTillMeeting}
\label{ZigZagTillMeetingAlgorithm}
\begin{algorithmic}[1]
\For{$i \gets 0$ to $\infty$ \textbf{until} the target is found}

\Statex \textbf{Sender's moves:}
  \State $S$ moves to the left a distance $x_{2i}$ unless the target is found;
  \If{the target is not found}
    \State $S$ reverses direction and moves back to meet $R$;
    \If{$S$ does not meet $R$ on the way}
       \State $S$ continues in the same direction until it catches the target;
    \Else
       \State $S$ reverses direction after meeting $R$ and sets its next travel distance to $x_{2i+2}$;
    \EndIf
  \EndIf
  \If{$S$ reaches the target}
     \State $S$ notifies $R$ that it has reached the target;
     \State $S$ stays with the target and moves with it;
  \EndIf

  \vspace{0.5em} 

\Statex \textbf{Receiver's moves:}
  \State $R$ moves to the right a distance $x_{2i+1}$ unless the target is found;
  \If{the target is not found}
    \State $R$ reverses direction and moves back to meet $S$;
    \If{$R$ does not meet $S$ on the way}
       \State $R$ continues in the same direction until it catches the target;
    \Else
       \State $R$ reverses direction after meeting $S$ and sets its next travel distance to $x_{2i+3}$;
    \EndIf
  \EndIf
  \If{$R$ reaches the target}
     \State $R$ stays with the target and moves with it;
  \EndIf
  \If{$R$ receives notification from $S$ that it has reached the target}
     \State $R$ reverses direction and proceeds to the target;
  \EndIf

\EndFor
\end{algorithmic}
\end{algorithm}

If neither robot finds the target during the first iteration, they proceed to the second iteration.  In this iteration, robot $S$ moves a distance of $x_{2}$, and robot $R$ moves a distance of $x_{3}$. This pattern continues for subsequent iterations $i$, where, in each iteration $i$, robot $S$ moves a distance of $x_{2i}$, and robot $R$ moves a distance of $x_{2i+1}$. We define a sequence of distances $x_{i}$ for $i \geq 0$ such that $x_{i} = a^{i+1}$, where $a$ is a non-negative real number. During each iteration $i$, the robots meet at predetermined points $y_{i}$. We interpret \( y_i \) as the position on the line where both robots meet during iteration \( i \). The origin is set at \( y_0 = 0 \).
The sequence of points $y_{i}$ is calculated based on the assumption that the robots meet at $y_{i-1}$ during iteration $i-1$ and plan to meet at $y_{i}$ during iteration $i$. This leads to the following recurrence relation: $x_{2i-1} - y_{i-1} + x_{2i-1} - y_{i} = y_{i-1} + 2x_{2i-2} + y_{i}$, which after simplification reduces to:
$y_{i} = x_{2i-1} - x_{2i-2} - y_{i-1}$.
The algorithm terminates when one of the robots finds the target. At that point, the other robot realizes the target has been found, as it does not encounter its counterpart at the expected intersection point. The second robot then proceeds directly toward the target. Although the sequence $y_i$ is not used in the proof below, it is worth mentioning that it represents the points at which both robots meet during each iteration, and can be useful to understand the algorithm’s behavior.
\begin{theorem}
\label{thm:ZigZagTillMeetingAlgorithm}    
The competitive ratio of Algorithm~\ref{ZigZagTillMeetingAlgorithm} is upper bounded by
\begin{align}
\label{eq:zigzag1}
    1 + \frac{2(a^5 + a^4)}{a^4 + v a^4 + 2a v + v - 1}.
\end{align}
where $a$ (as a function of $v$) is the positive root of the following equation
\begin{align}
\label{eq:solution}
    (1 + v) a^5 + 8v a^2 + (11v - 5) a + 4(v - 1)=0.
\end{align}
\end{theorem}
\begin{proof} 
The main idea in the algorithm is to have both robots cover incremental sequences of distances $x_i$, where $S$ covers the even terms and $R$ covers the odd terms. It is noticeable that the worst-case scenario occurs if $R$ just misses the target at some iteration $i$. The worst-case competitive ratio would remain the same whether the target is missed at the odd term $x_{4n-1}$ or at any other odd term, as the logic for calculating the competitive ratio is the same in both cases.

To explain the significance of the term $x_{4n-1}$, we observe a pattern in the movement. Initially, $S$ moves to the left by distance $x_0$ and returns to meet $R$, which then moves a distance of $x_3$. Afterward, $R$ returns to meet $S$, who then moves a distance of $x_4$, and so on. In this process, $R$ next moves $x_7$, then returns, and the cycle continues. Thus, the total time spent in this alternating movement can be represented by the sum $x_0 + x_3 + x_4 + x_7 + \dots$, where the terms $x_3, x_7, \dots$ are of the form $x_{4n-1}$.

To compute the competitive ratio, assume that the target is missed by $R$ after covering a distance $x_i$, where $x_i$ belongs to the sequence $x_{4n-1}$ at some iteration indexed by $\frac{i+1}{2}$. In this case, the competitive ratio is given by:
\begin{align*}
CR &=  \frac{2\sum_{j=0}^{\frac{i-3}{4}} x_{4j} + 2\sum_{j=0}^{\frac{i-3}{4}} x_{4j+3} + 2x_{i+1} + \frac{d - 2v \sum_{j=0}^{\frac{i-3}{4}} x_{4j} - 2v \sum_{j=0}^{\frac{i-3}{4}} x_{4j+3} - 2v x_{i+1}}{1+v}}{\frac{d}{1+v}} \\
&= 1 + \frac{2 \sum_{j=0}^{\frac{i-3}{4}} x_{4j} + 2 \sum_{j=0}^{\frac{i-3}{4}} x_{4j+3} + 2x_{i+1}}{d}
\end{align*}

If the target is just missed at iteration $i$, then we have the following:
\begin{align*}
   d - 2v \sum_{j=0}^{\frac{i-3}{4}} x_{4j} - 2v \sum_{j=0}^{\frac{i-7}{4}} x_{4j+3}-vx_{i} \geq a^i
\end{align*}
Since
$
\sum_{j=0}^{\frac{i-3}{4}} x_{4j} = \frac{a^{i+1} - 1}{a^4 - 1} \mbox{ and } \sum_{j=0}^{\frac{i-7}{4}} x_{4j+3} = \frac{a^{i} - a^3}{a^4 - 1}
$
we end up having the following:
\begin{align*}
    a^i \leq \frac{d(a^4 - 1)}{a^4 + v a^4 + 2a v + v - 1} + \frac{2v a^3 + 2v}{a^4 + v a^4 + 2a v + v - 1}
\end{align*}
The competitive ratio satisfies
\begin{align*}
    CR &= 1 + \frac{2a^{i+5} + 2a^{i+4} - 2a^3 - 2}{d(a^4 - 1)} 
    \leq 1 + 2a^i \left( \frac{a^5 + a^4}{d(a^4 - 1)} \right) \\
    &\leq 1 + 2 \left( \frac{d(a^4 - 1) + 2v a^3 + 2v}{a^4 + v a^4 + 2a v + v - 1} \right) \left( \frac{a^5 + a^4}{d(a^4 - 1)} \right) \\
    &\leq 1 + \frac{2(a^5 + a^4)}{a^4 + v a^4 + 2a v + v - 1},
\end{align*}
which is Inequality~\eqref{eq:zigzag1}, where the last inequality is derived from the previous one by taking the limit as $d$ goes to infinity. The derivative of $f(a) = \frac{2(a^5 + a^4)}{a^4 + v a^4 + 2a v + v - 1}$ is
\begin{align*}
   f'(a) &= \frac{(1 + v) a^5 + 8v a^2 + (11v - 5) a + 4(v - 1)}{(a^4 + v a^4 + 2a v + v - 1)^2}.
\end{align*}
The optimizer can be found by setting $f'(a) = 0$. This is easily seen to be equivalent to $(1 + v) a^5 + 8v a^2 + (11v - 5) a + 4(v - 1)=0$, which is Equation~\eqref{eq:solution} above. 
This proves the theorem. From Equation~\eqref{eq:solution}, we observe that the condition $a > 1$ is satisfied if and only if $v < \frac{1}{3}$. Comparing this algorithm with Algorithm~\ref{NoDistanceTowardOppositeDirectionAlgorithm}, it is clear that Algorithm~\ref{NoDistanceTowardOppositeDirectionAlgorithm} performs better for any value of $v \leq \frac{1}{3}$.
\end{proof}
 
\subsubsection{Waiting Algorithm.}

The waiting algorithm is the strategy whereby both robots remain at the origin (indefinitely). This algorithm is easy to analyze and a proof can be found in \cite{coleman2023line}.
\begin{theorem}
When the target moves toward the origin, the competitive ratio for the waiting algorithm is $1+\frac{1}{v}$.

\end{theorem}

Figure~\ref{Performance} illustrates and compares the upper bounds of the two algorithms as a function of $v$.

\begin{figure}[h]
\begin{center}
\includegraphics[width=11cm,height=7cm]{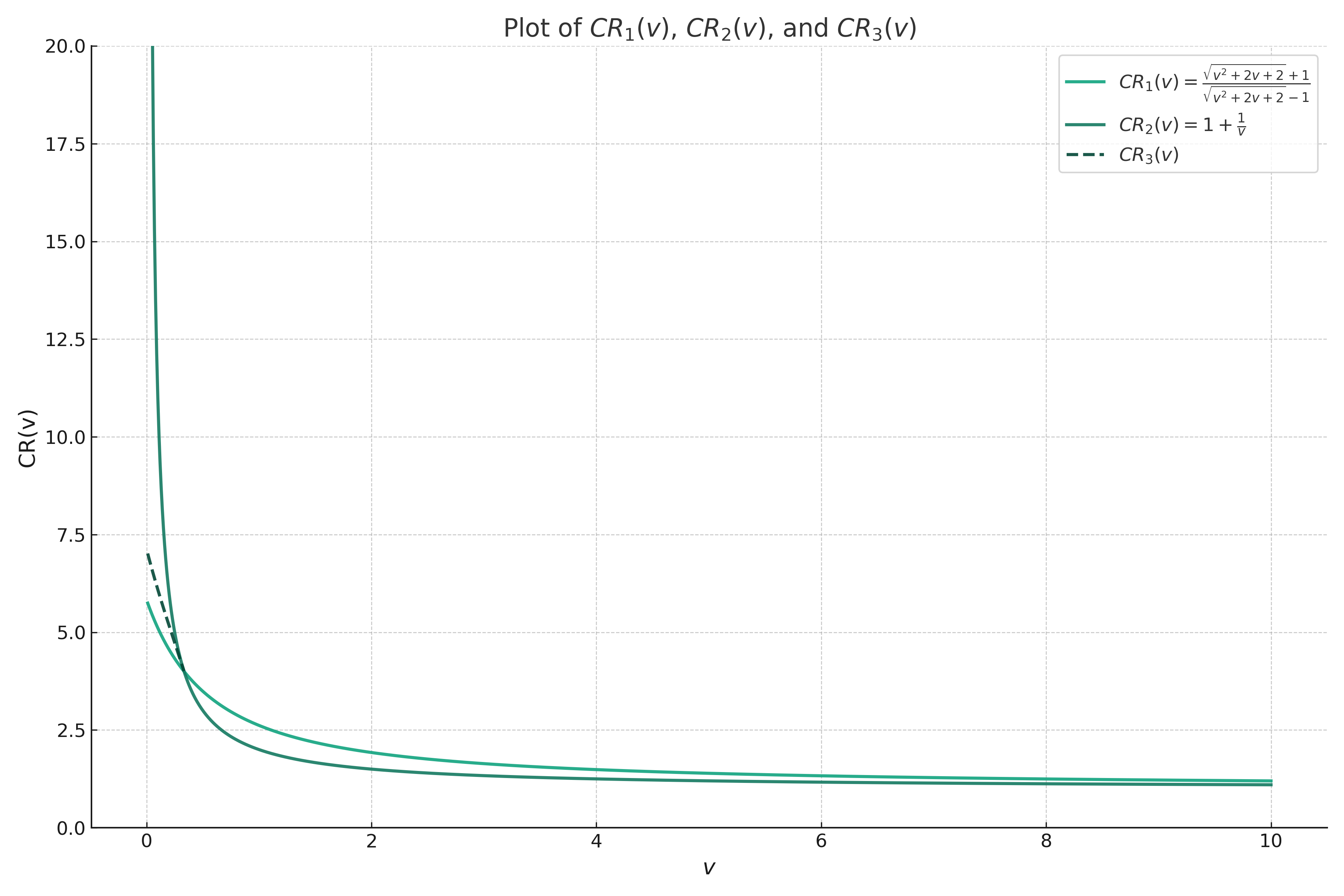}
\end{center}
\caption[Graph showing the performance of the three algorithms, illustrating how the competitive ratio fluctuates with the speed of the moving target]{Algorithm~\ref{NoDistanceTowardOppositeDirectionAlgorithm} achieves the best performance for $0\leq v\leq \frac{1}{3}$, while Algorithm~\ref{ZigZagTillMeetingAlgorithm} outperforms the waiting algorithm within that same interval. For $v\geq \frac{1}{3}$
 , the waiting algorithm becomes the most competitive. Interestingly, all three algorithms share the same competitive ratio at $v=\frac{1}{3}$}
\label{Performance}
\end{figure}
\subsection{Target moving away from the origin}
The target moves away from the origin. We have no knowledge of the initial distance of the target from the origin. We have knowledge only of the speed of the moving target. Assume that $R$ moves in one direction with maximum unit speed and $S$ moves in the other direction with speed $u$. There are two cases to consider. Either $S$ or $R$ finds the target first.  If $S$ finds the target first, it informs $R$ which proceeds to the target, otherwise if $R$ does, it switches its direction to reach $S$, and then both robots move to the target with unit speed.
\begin{algorithm}[H]
\caption{NoDistanceAway S/R} 
\label{NoDistanceAwayS/RAlgorithm}
\begin{algorithmic}[1]
\State {$S$ moves to the left with speed $u$; $R$ moves to the right with speed $1$}
\If {$R$ finds the target first}
\State{it changes direction and catches up to $S$;}
\State{The two robots change direction and move towards the target with unit speed;}
\Else
\If {$S$ finds the target first} 
\State{it informs $R$ that the target has been found and moves with the target;}
\State{$R$ moves with unit speed to the target;}
\EndIf
\EndIf
\end{algorithmic}
\end{algorithm}
We prove the following result.
\begin{theorem}
\label{theorem1}
The competitive ratio of Algorithm~\ref{NoDistanceAwayS/RAlgorithm} is 
\begin{equation}
\label{eq:awalg}
\frac{\sqrt{v^2-2v+2}+1}{\sqrt{v^2-2v+2}-1}.
\end{equation}
\end{theorem}
\begin{proof} (Theorem~\ref{theorem1})
$S$ moves with speed $u$ and $R$ with unit speed then we have two cases to consider depending on which of $S$ or $R$ finds the target first.
 \begin{itemize}
     \item {\bf $S$ finds the target first.} 
     The competitive ratio is as follows:
     \begin{align}
         CR&=\label{eq:finds0} \frac{\frac{d}{u-v}}{\frac{d}{1-v}}+\frac{\frac{\frac{d}{u-v}+\frac{du}{u-v}}{1-v}}{\frac{d}{1-v}}
         =\frac{\frac{d}{u-v}+\frac{du+d}{(1-v)(u-v)}}{\frac{d}{1-v}}
         =\frac{2-v+u}{u-v}.
     \end{align}
     \item {\bf $R$ finds the target first.}
     To reach the target, it needs time: 
$\frac{d}{1-v}$.
At this time, $S$ would be away from $R$ by distance $\frac{d}{1-v}+\frac{du}{1-v}$, thus for $R$ to capture $S$, it needs time:
\begin{equation}
        \label{eq5}
       \frac{\frac{d}{1-v}+\frac{du}{1-v}}{1-u}=\frac{d+du}{(1-v)(1-u)} .
\end{equation}
Thus both robots would be away from the target by $\frac{d+du}{(1-v)(1-u)}+\frac{dv+dvu}{(1-v)(1-u)}$. It follows that for both robots to catch the target, they need time:
\begin{equation}
        \label{eq6}
       \frac{d+du+dv+dvu}{(1-v)^2(1-u)}.
\end{equation}
Combining $\frac{d}{1-v}$ with equations \eqref{eq5},~\eqref{eq6}, the competitive ratio becomes:
\begin{align}
    \frac{\frac{d}{1-v}+\frac{d+du}{(1-v)(1-u)}+\frac{d+du+dv+dvu}{(1-v)^2(1-u)}}{\frac{d}{1-v}}
   &=\label{eq:finds1} \frac{3-v+u+vu}{(1-v)(1-u)} .
\end{align}
 \end{itemize}
 It is required to find the optimal value $u$. Setting the competitive ratios in Equations~\eqref{eq:finds0}~and~\eqref{eq:finds1} to be equal, we get
 $\frac{3+u+uv-v}{(1-v)(1-u)}=\frac{2-v+u}{u-v}.$
 We conclude that $2u^{2}+4u-4uv-2=0$,  thus $u=(v-1)+\sqrt{v^2-2v+2}$. Substituting $u$ above yields Equation~\eqref{eq:awalg}.
This proves the theorem.~\qed 
\end{proof}

\section{The NoSpeed Model}
\label{sec: The NoSpeed Model}
In this section, we consider the NoSpeed model ($d$ is known but $v$ is not known to the robots) and distinguish the cases where the target is moving away from or toward the origin.

\subsection{Target moving toward the origin}

Consider the following algorithm. 
 \begin{algorithm}[H]
\caption{NoSpeedToward S/R} 
\label{NoSpeedTowardS/RAlgorithm}
Input: Target Initial distance d
\begin{algorithmic}[1]
\State {The two robots select the same direction and move together with unit speed, distance $d$}
\If {Target is not found}
\State{Both robots change direction and move together until they encounter the target, at which point they stop.}
\EndIf
\end{algorithmic}
\end{algorithm}

\begin{theorem}
\label{NSTowardS/RTheorem}
The competitive ratio of Algorithm~\ref{NoSpeedTowardS/RAlgorithm} is upper bounded by $3$ and this is optimal.
\end{theorem}
\begin{proof}
     In the worst case scenario, if the target’s speed is very small and one of the robots moves in the opposite direction to where the target is located, the robot will end up at a distance of approximately
$2d$ from the target. At that point, the robot will be separated from the target by a distance $2d-dv$.
  Thus, the competitive ratio is
$
    CR =\frac{d+\frac{2d-dv}{1+v}}{\frac{d}{1+v}}=3 .
$
 This proves the upper bound. 
 To prove the lower bound, we argue as follows. One of the two points $-d, d$ must be visited by either $S$ or $R$, otherwise if the speed of the target is too small, then none of the robots would be able to capture the target. There are two cases to consider, either $S$ or $R$ reaches one of the two points $-d, d$. Let us consider the first case, Assuming that $S$ reaches point $d$, the adversary can initially place the target at point $-d$. Thus for the robots to capture the target, considering that the speed of the target is $v$, the competitive ratio is at least:
$
    \frac{d+\frac{2d-dv}{1+v}}{\frac{d}{1+v}}=3.
$    
For the second case, assuming that $R$ reaches point $d$, then the adversary can initially place the target at point $-d$ and we end up having a similar competitive ratio $3$. 
This proves the theorem.
~\qed
\end{proof}
\subsection{Target moving away from the origin}

 The target moves away from the origin in any direction. Let us consider a monotone increasing sequence $\{f_{i}:i\geq 0\}$ of non-negative integers. The idea is to guess the speed of the target. 
\begin{algorithm}[H]
\small
\caption{NoSpeedAway S/R}
\label{NoSpeedAwayS/RAlgorithm}
\textbf{Input:} Target's initial distance $d$ \\
\textbf{Given:} Increasing integer sequence $f_{i}$ with $f_0 = 1$, $f_i < f_{i+1}$; initial times $t_1 = t_2 = 0$
\begin{algorithmic}[1]
\For{$i \gets 0$ to $\infty$ \textbf{until} the target is found}

\Statex \textbf{Sender's (S) moves:}
  \State Set $v_{2i} = 1 - 2^{-f_{2i}}$, $x_{2i} = -\frac{d + t_1 v_{2i}}{1 - v_{2i}}$
  \State $S$ moves left by $x_{2i}$ and returns to meet $R$
  \If{$S$ does not meet $R$}
    \State $S$ continues to pursue the target
  \Else
    \State $S$ reverses and sets next guess $v_{2i+2}$
  \EndIf
  \If{$S$ finds the target}
    \State $S$ stays with the target
  \EndIf

\vspace{0.5em}

\Statex \textbf{Receiver's (R) moves:}
  \State Set $v_{2i+1} = 1 - 2^{-f_{2i+1}}$, $x_{2i+1} = \frac{d + t_2 v_{2i+1}}{1 - v_{2i+1}}$
  \State $R$ moves right by $x_{2i+1}$ and returns to meet $S$
  \If{$R$ misses $S$}
    \State $R$ continues to pursue the target
  \Else
    \State $R$ continues with $S$ using $v_{2i+1}$
    \If{target not found}
      \State $R$ reverses and updates to $v_{2i+3}$
    \EndIf
  \EndIf
  \If{$R$ finds the target}
    \State $R$ stays with the target
  \EndIf

\vspace{0.5em}

\Statex \textbf{Time Updates:}
  \State $t_1 \gets t_1 + 2|x_{2i}| + 2 \times$(meeting time with $R$)
  \State $t_2 \gets t_2 + 2|x_{2i+1}| + 2 \times$(pursuit time with $S$ and return)

\EndFor
\end{algorithmic}
\end{algorithm}

The algorithm works through iterations. We will use the guess $v_{i}=1-2^{-f_{i}}$. At the beginning, both robots $S$ and $R$ are situated at the origin. During each iteration $i$, robot $S$ moves in one direction with a guess of $v_{2i}$, and $R$ moves in the other direction with a guess of $v_{2i+1}$. If $S$ finds the target first, it communicates with $R$ to proceed to the target. Otherwise if $R$ finds the target first, it moves with the target and waits for $S$ to proceed to the target. If the target is not found by neither $S$ nor $R$, then each robot returns back until they meet. At the meeting point, $R$ would go back with $S$ to confirm whether the target is caught with a guess of $v_{2i+1}$ on the side where $S$ was moving. If the target is not found by $S$ and $R$, then $R$ reverses its direction and proceeds with a guess of $v_{2i+3}$, while $S$ continues in the same direction with a guess of $v_{2i+2}$. The competitive ratio is as follows.
\begin{theorem}
\label{thm:NoSpeedAwayS/RAlgorithm}
The competitive ratio of Algorithm~\ref{NoSpeedAwayS/RAlgorithm} is 
$
1 + O(u^{\frac{10}{3}}  \log u),
$
where $u = \frac{1}{1-v}$.
\end{theorem}

\begin{proof} 
   Assume that $x_{i}$ represents the time required by $R$ to move in one direction with a guess of the speed $v_{i}$ and then return to the origin. Furthermore, it includes the time needed to move in the opposite direction after meeting $S$ and then return to the origin again, all while using the same speed guess $v_{i}$. Thus, at iteration $\frac{i+1}{2}$, $x_{i}$ can be expressed as follows:
   \begin{align*}
       x_{i}&=2\cdot 2^{f_{i}}\cdot (d+v_{i}\sum_{j=0}^{\frac{i-3}{2}}x_{2j+1})+2\cdot 2^{f_{i}}(2.2^{f_{i}}\cdot (d+v_{i}\sum_{j=0}^{\frac{i-3}{2}}x_{2j+1}))+ 2\cdot 2^{f_{i}}d\\
       &=(2 \cdot 2^{f_{i}}+1)(2.2^{f_{i}}(d+v_{i}\sum_{j=0}^{\frac{i-3}{2}}x_{2j+1})))+ 2\cdot 2^{f_{i}}d\\
       &=2\cdot 2^{f_{i}}.d+2\cdot 2^{f_{i}}\cdot v_{i}\sum_{j=0}^{\frac{i-3}{2}}x_{2j+1}+4\cdot 2^{f_{i+1}}\cdot d+4\cdot 2^{f_{i+1}}v_{i}\sum_{j=0}^{\frac{i-3}{2}}x_{2j+1}+ 2\cdot 2^{f_{i}}d .
   \end{align*}
   Since $v_{i}=1-2^{-f_{i}}$, we get for $x_i$:
   \begin{align*}
       &x_{i}=4\cdot 2^{f_{i}}d+4\cdot 2^{f_{i+1}}d-2\cdot 2^{f_{i}}\sum_{j=0}^{\frac{i-3}{2}}x_{2j+1}-2\sum_{j=0}^{\frac{i-3}{2}}x_{2j+1}+4\cdot 2^{f_{i+1}}\sum_{j=0}^{\frac{i-3}{2}}x_{2j+1} .
   \end{align*}  
   We conclude the following:
   \begin{align*}
       &\sum_{j=0}^{\frac{i-3}{2}}x_{2j+1}=\frac{x_{i}-4\cdot d\cdot 2^{f_{i}}-4\cdot d\cdot 2^{f_{i+1}}}{4\cdot 2^{f_{i+1}}-2\cdot 2^{f_{i}}-2}\\
       &\implies \sum_{j=0}^{\frac{i-1}{2}}x_{2j+1}=\frac{x_{i+2}-4\cdot d\cdot 2^{f_{i+2}}-4\cdot d\cdot 2^{f_{i+3}}}{4\cdot 2^{f_{i+3}}-2\cdot 2^{f_{i+2}}-2}\\
        &\implies x_{i}=\frac{x_{i+2}-4\cdot d\cdot 2^{f_{i+2}}-4\cdot d\cdot 2^{f_{i+3}}}{4\cdot 2^{f_{i+3}}-2\cdot 2^{f_{i+2}}-2}-\frac{x_{i}-4\cdot d\cdot 2^{f_{i}}-4\cdot d\cdot 2^{f_{i+1}}}{4\cdot 2^{f_{i+1}}-2\cdot 2^{f_{i}}-2}.
   \end{align*}
   We deduce the following:
   \begin{align*}
       \tfrac{x_{i+2}}{4\cdot 2^{f_{i+3}}-2\cdot 2^{f_{i+1}}-2}&\leq (1+\tfrac{1}{4\cdot 2^{f_{i+1}}-2\cdot 2^{f_{i}}-2})x_{i}+(\tfrac{4\cdot 2^{f_{i+2}}d+4\cdot 2^{f_{i+3}}d}{4\cdot 2^{f_{i+3}}-2\cdot 2^{f_{i+2}}-2}-\tfrac{4\cdot 2^{f_{i}}d+4\cdot 2^{f_{i+1}}d}{4\cdot 2^{f_{i+1}}-2\cdot 2^{f_{i}}-2})\\
       &\leq (1+\tfrac{1}{2\cdot 2^{f_{i+1}}})x_{i} \mbox{ (since the other term is less than 0)}\\
       &\implies x_{i+2}\leq 4\cdot 2^{f_{i+3}}(1+\tfrac{1}{2\cdot 2^{f_{i+1}}})x_{i}\\
       &\implies x_{i+2}\leq 4\cdot 2^{f_{i+3}}x_{i}+2\cdot 2^{f_{i+3}-f_{i+1}}x_{i}\leq 6\cdot 2^{f_{i+3}}x_{i}\\
        &\implies x_{i+2}\leq 2^{\sum_{j=0}^{\frac{i+2}{2}}f_{2j+1}}\cdot 6^{\frac{i+1}{2}}\cdot 4d .
   \end{align*}
   The capture time $T$ would be as follows:
   \begin{align*}
       &\sum_{j=0}^{\tfrac{i-3}{2}}x_{2j+1}+2\cdot (2^{f_{i}}(d+v_{i}\cdot \sum_{j=0}^{\tfrac{i-3}{2}}x_{2j+1}))+\tfrac{d}{1-v}+\tfrac{\sum_{j=0}^{\tfrac{i-3}{2}}x_{2j+1}v}{1-v}+\tfrac{2\cdot2^{f_{i}}\cdot dv}{1-v}+\tfrac{2\cdot 2^{f_{i}}\cdot v_{i}\cdot \sum_{j=0}^{\tfrac{i-3}{2}}x_{2j+1}v}{1-v}\\
       \end{align*}
       Since $v_{i}=1-2^{-f_{i}}$, then we get:
       \begin{align*}
       T&\leq \frac{2\cdot 2^{f_{i}}d+2\cdot 2^{f_{i}}\sum_{j=0}^{\frac{i-3}{2}}x_{2j+1}-2\sum_{j=0}^{\frac{i-3}{2}}x_{2j+1}+d+\sum_{j=0}^{\frac{i-3}{2}}x_{2j+1}}{1-v}\\
       &\leq \frac{2\cdot 2^{f_{i}}\cdot\sum_{j=0}^{\frac{i-3}{2}}x_{2j+1}-\sum_{j=0}^{\frac{i-3}{2}}x_{2j+1}+d+2\cdot 2^{f_{i}}d}{1-v} .
   \end{align*}
   We conclude that the competitive ratio would be as follows:
   \begin{align*}
       CR&=2\cdot 2^{f_{i}}+1+\frac{2\cdot 2^{f_{i}}\cdot \sum_{j=0}^{\frac{i-3}{2}}x_{2j+1}}{d}-\frac{ \sum_{j=0}^{\frac{i-3}{2}}x_{2j+1}}{d}.
   \end{align*}
  Considering that the target is found in iteration $i$, we have the condition $2^{f_{i-1}} \leq \frac{1}{1 - v}$. Assuming that $f_i = 2^i$, this leads to the following inequality:
   \begin{equation}
        \label{equation9}
        \sum_{j=0}^{\frac{i}{2}}f_{2j+1}=\sum_{j=0}^{\frac{i}{2}}4^j\cdot 2=2(\frac{4^{\frac{i}{2}+1}-1}{3})=8\cdot \frac{4^{\frac{i}{2}}}{3}-\frac{2}{3}=\frac{8\cdot 2^{i}}{3}-\frac{2}{3}
   \end{equation}
    In addition to that we have:  
\begin{align}
2^{f_{i}}\sum_{j=0}^{\frac{i-3}{2}}x_{2j+1}&\leq \frac{2^{f_{i}}\cdot x_{i}}{2\cdot 2^{f_{i+1}}}=\frac{x_{i}}{2\cdot 2^{f_{i}}}\leq \frac{2^{\sum_{j=0}^{\frac{i}{2}}f_{2j+1}}\cdot 6^{\frac{i-1}{2}}\cdot 4d}{2\cdot 2^{f_{i}}} \notag \\
&\leq 2^{\tfrac{8}{3}\cdot 2^{i}-2^{i}}\cdot 6^{\tfrac{i-1}{2}}\cdot d\cdot 2^{\tfrac{1}{3}}\leq 2^{\tfrac{5}{3}\cdot 2^{i}}6^{\tfrac{i-1}{2}}\cdot d\cdot 2^{\tfrac{1}{3}}\notag \\
&\leq \left(\tfrac{1}{1-v}\right)^{\tfrac{10}{3}}\cdot d \cdot 6^{\tfrac{\log(\log(\tfrac{1}{1-v}))}{2}}\cdot 2^{\tfrac{1}{3}}\label{equation10}
\end{align}

We also have the following:
   \begin{align}
\sum_{j=0}^{\frac{i-3}{2}}x_{2j+1}&=\frac{x_{i}}{2\cdot 2^{f_{i+1}}}\leq \frac{2^{\sum_{j=0}^{\frac{i}{2}}f_{2j+1}}\cdot 6^{\frac{i-1}{2}}\cdot 4d}{2\cdot 2^{2^{i+1}}}\leq \frac{2^{\frac{8}{3}\cdot 2^{i}}\cdot 6^{\frac{i-1}{2}}\cdot d\cdot 2^{\frac{1}{3}}}{2^{2^{i+1}}} \notag \\
&\leq 2^{\frac{2}{3}\cdot 2^{i}}\cdot 6^{\frac{i-1}{2}}\cdot d\cdot 2^{\frac{1}{3}}\leq \left(\frac{1}{1-v}\right)^{\frac{4}{3}}\cdot d\cdot 2^{\frac{1}{3}}\cdot 6^{\frac{\log(\log(\frac{1}{1-v}))}{2}} \label{equation11}
\end{align}
Combining Equations~\eqref{equation9},~\eqref{equation10},~\eqref{equation11}, with the fact that $2^{f_{i}}\leq (\frac{1}{1-v})^2$, the competitive ratio satisfies
\begin{align*}
    CR&\leq 1+\frac{2}{(1-v)^2}-\frac{1}{(1-v)^{\frac43}}\cdot 6^{\frac{\log\log\frac{1}{1-v}}{2}} \cdot 2^{\frac{1}{3}}+ \frac{1}{(1-v)^{\frac{10}{3} }} 
    \cdot 6^{\frac{\log\log\frac{1}{1-v}}{2}}\cdot 2^{\frac{1}{3}}\\
    &\leq 1+O\left(\frac{1}{(1-v)^2}\right)+O\left(\frac{\log \frac{1}{1-v}}{(1-v)^{\frac{10}{3}}}  \right).
\end{align*}
This proves the theorem.~\qed
\end{proof}
\begin{theorem}
\label{thm:NoSpeedAwayS/RAlgorithmLB}
The competitive ratio of the NoSpeedAway S/R model is at least 
$
1 + \Omega\left(\frac{1}{1 - v}\right).
$
\end{theorem}
\begin{proof} 
Consider an algorithm $A$ solving the capturing problem for two robots. Let $d$ be the starting distance of the target from the origin and $v<1$ its speed. For any $\epsilon > 0$ but arbitrarily small consider the first time one of the two robots reaches a distance $\frac d{1-v} - \epsilon$ away from the origin, following algorithm $A$, in either direction from the origin of the line. Let's call this the first robot. The adversary places the mobile in the opposite direction on the line. Regardless of where the second of the two robots is, the first robot will be distance $\frac{2d}{1-v} - 2\epsilon$ away from current position of the mobile. Therefore, the first robot needs time $\frac{d}{1-v} - \epsilon$ plus additional time 
\begin{equation}
\label{eq:srlb1}
\frac{\frac{2d}{1-v} - 2\epsilon}{1-v} = \frac{2d}{(1-v)^2}- \frac{2\epsilon}{1-v}.
\end{equation}
to move on the other side of the line where the mobile is. Consolidating the terms and simplifying Equation~\eqref{eq:srlb1} above we see that the total time for the first robot to reach the mobile (which is moving on the other side of the line) will be at least 
\begin{equation}
\label{eq:srlb2}
\frac{d-2\epsilon}{1-v} - \epsilon + \frac{2d}{(1-v)^2}.
\end{equation}
As a consequence, the resulting competitive ratio of algorithm $A$ must satisfy
\begin{equation}
\label{eq:srlb3}
CR \geq \frac{\frac{d-2\epsilon}{1-v} - \epsilon + \frac{2d}{(1-v)^2}}{\frac{d}{1-v}} .
\end{equation}
Clearly, for $d, v$ constants, the righthand side of Inequality~\eqref{eq:srlb3} converges to $1+ \frac{2}{1-v}$, as $\epsilon \to 0$, which proves the lower bound.
~\qed
\end{proof}
\section{The NoKnowledge Model}
\label{sec: The NoKnowledge Model}

In this section, we consider the NoKnowledge model (neither $d$ nor $v$ is known to the robots) and distinguish the cases where the target is moving away from or toward the origin.

\subsection{Target moving toward the origin}
\begin{theorem}
\label{thm:Target moving toward the origin}
    The optimal competitive ratio is $1+\frac{1}{v}$ and is given by the waiting algorithm. 
\end{theorem}
\begin{proof} 
The upper bound is well-known~\cite{coleman2023line}.
    For the lower bound, consider an algorithm where the robot does not wait and instead moves a distance $x>0$ in a certain direction, after waiting at the origin for time $t>0$. Consider the scenario where the target with speed $v=\frac{d}{t+x}$ is at distance $d$ away from the origin in the opposite direction. Thus, the target reaches the origin at the time the robot reaches $x$. The earliest meeting point is:
        $t+x+\frac{x}{1+v}=\frac{d}{v}+\frac{x}{1+v}\geq \frac{d}{v} $.
    It follows that the competitive ratio is at least $\frac{\frac{d}{v}}{\frac{d}{1+v}}=1+\frac{1}{v}$.
    This completes the proof of the theorem.~\qed
\end{proof}

\subsection{Target moving away from the origin}

The target moves away from the origin in either direction.  Let us consider two monotone increasing sequences of non-negative integers $\{f_{i}, g_{i} : i \geq 0\}$. The idea is to try to guess the speed and the initial distance of the target. The algorithm works through iterations. We will use the guesses $v_{i} = 1 - 2^{-f_{i}}$ and $d_{i} = 2^{g_{i}}$. At the beginning, both robots $S$ and $R$ are situated at the origin.

 During each iteration $i$, robot $S$ moves in one direction with a guess of $v_{2i}$ for the speed and $d_{2i}$ for the distance. $R$ moves in the other direction with a guess of $v_{2i+1}$ for the speed and $d_{2i+1}$ for the distance. If $S$ finds the target first, it communicates with $R$ to proceed to the target. Otherwise, if $R$ finds the target first, it moves with the target and waits for $S$ to proceed to the target.

If the target is not found by either $S$ or $R$, then each robot returns back until they meet. At the meeting point, $R$ would go back with $S$ to confirm whether the target is caught with a guess of $v_{2i+1}$ for the speed and $d_{2i+1}$ for the distance on the side where $S$ was moving. If the target is not found by $S$ and $R$, then $R$ reverses its direction and proceeds with a guess of $v_{2i+3}$ for the speed and $d_{2i+3}$ for the distance, while $S$ continues in the same direction with a guess of $v_{2i+2}$ for the speed and $d_{2i+2}$ for the distance. 

\begin{algorithm}[H]
\small
\caption{NoKnowledgeAway S/R}
\label{NoKnowledgeAwayS/RAlgorithm}
\begin{algorithmic}[1]
\State \textbf{Given:} Increasing integer sequences $f_i$, $g_i$ with $f_0 = 1$, $g_0 = 0$, $f_i < f_{i+1}$, $g_i < g_{i+1}$, and initial times $t_{1} = t_{2} = 0$
\For{$i \gets 0$ to $\infty$ until target is found} 
  \State $S$: Set $v_{2i} = 1 - 2^{-f_{2i}}$, $d_{2i} = 2^{g_{2i}}$, $x_{2i} = \frac{d_i + t_1 v_{2i}}{1 - v_{2i}}$
  \State Move left, then return to meet $R$
  \If{no meeting with $R$}
    \State Continue to pursue target
  \Else
    \State Reverse and update to $v_{2i+2}, d_{2i+2}$
  \EndIf
  \If{target found} \State Move with target \EndIf

  \State $R$: Set $v_{2i+1} = 1 - 2^{-f_{2i+1}}$, $d_{2i+1} = 2^{g_{2i+1}}$, $x_{2i+1} = \frac{d_i + t_2 v_{2i+1}}{1 - v_{2i+1}}$
  \State Move right, then return to meet $S$
  \If{no meeting with $S$}
    \State Continue to pursue target
  \Else
    \State Continue with $S$ using $v_{2i+1}, d_{2i+1}$
    \If{target not found}
      \State Reverse and update to $v_{2i+3}, d_{2i+3}$
    \EndIf
  \EndIf
  \If{target found} \State Move with target \EndIf

  \State $t_1 \gets t_1 + 2|x_{2i}| + 2 \times$(round-trip time with $R$)
  \State $t_2 \gets t_2 + 2|x_{2i+1}| + 2 \times$(time with $S$ for guess coverage and return)
\EndFor
\end{algorithmic}
\end{algorithm}

\begin{theorem}
\label{thm:NoKnowledgeAwayS/RAlgorithm}
    The competitive ratio of Algorithm~\ref{NoKnowledgeAwayS/RAlgorithm} is upper bounded by 
$    
1 + O\left( M^{\frac{16}{3}} \log(M) \log \log^{\frac{3}{2}} M \right)
$    
where $M = \max\{d,\frac{1}{1-v}\}$.    
\end{theorem}
\begin{proof} 
   Assume that $x_{i}$ represents the time required by $R$ to move to the right with a guess for the speed $v_{i}$ and the initial distance $d_{i}$, then return to the origin. Additionally, it includes the time required to move to the left after meeting $S$ and return to the origin again, both using the same guesses for the speed $v_{i}$ and the initial distance $d_{i}$. Thus, at iteration $i$, $x_{i}$ can be expressed as follows:
   \begin{align*}
       x_{i} &= 2 \cdot 2^{f_{i}} \cdot \left(2^{g_{i}} + v_{i} \sum_{j=0}^{\frac{i-3}{2}} x_{2j+1}\right) 
       + 2 \cdot 2^{f_{i}} \cdot \left(2 \cdot 2^{f_{i}} \cdot \left(2^{g_{i}} + v_{i} \sum_{j=0}^{\frac{i-3}{2}} x_{2j+1}\right)\right) 
       + 2 \cdot 2^{f_{i}} \cdot 2^{g_{i}} \\
       &= \left(2 \cdot 2^{f_{i}} + 1\right) \cdot \left(2 \cdot 2^{f_{i}} \cdot \left(2^{g_{i}} + v_{i} \sum_{j=0}^{\frac{i-3}{2}} x_{2j+1}\right)\right) 
       + 2 \cdot 2^{f_{i}} \cdot 2^{g_{i}} \\
       &= 2 \cdot 2^{f_{i}} \cdot 2^{g_{i}} + 2 \cdot 2^{f_{i}} \cdot v_{i} \sum_{j=0}^{\frac{i-3}{2}} x_{2j+1} 
       + 4 \cdot 2^{f_{i+1}} \cdot 2^{g_{i}} + 4 \cdot 2^{f_{i+1}} \cdot v_{i} \sum_{j=0}^{\frac{i-3}{2}} x_{2j+1} 
       + 2 \cdot 2^{f_{i}} \cdot 2^{g_{i}}.
   \end{align*}
   
   Since $v_{i} = 1 - 2^{-f_{i}}$, we get for $x_i$
   \begin{align*}
       x_{i} &= 4 \cdot 2^{f_{i}} \cdot 2^{g_{i}} + 4 \cdot 2^{f_{i+1}} \cdot 2^{g_{i}} 
       - 2 \cdot 2^{f_{i}} \sum_{j=0}^{\frac{i-3}{2}} x_{2j+1} 
       - 2 \sum_{j=0}^{\frac{i-3}{2}} x_{2j+1} 
       + 4 \cdot 2^{f_{i+1}} \sum_{j=0}^{\frac{i-3}{2}} x_{2j+1}.
   \end{align*}
   We conclude the following:
   \begin{align*}
       \sum_{j=0}^{\frac{i-3}{2}} x_{2j+1} &= \frac{x_{i} - 4 \cdot 2^{g_{i}} \cdot 2^{f_{i}} - 4 \cdot 2^{g_{i}} \cdot 2^{f_{i+1}}}{4 \cdot 2^{f_{i+1}} - 2 \cdot 2^{f_{i}} - 2}
          \end{align*}

       Replacing $i$ with $i+2$, we get:
        \begin{align*}
       \sum_{j=0}^{\frac{i-1}{2}} x_{2j+1} &= \frac{x_{i+2} - 4 \cdot 2^{g_{i+2}} \cdot 2^{f_{i+2}} - 4 \cdot 2^{g_{i+2}} \cdot 2^{f_{i+3}}}{4 \cdot 2^{f_{i+3}} - 2 \cdot 2^{f_{i+2}} - 2}, 
          \end{align*}

        \begin{align*}
         x_{i} &= \sum_{j=0}^{\frac{i-1}{2}} x_{2j+1}- \sum_{j=0}^{\frac{i-3}{2}} x_{2j+1}\\
       &= \frac{x_{i+2} - 4 \cdot 2^{g_{i+2}} \cdot 2^{f_{i+2}} - 4 \cdot 2^{g_{i+2}} \cdot 2^{f_{i+3}}}{4 \cdot 2^{f_{i+3}} - 2 \cdot 2^{f_{i+2}} - 2} 
       - \frac{x_{i} - 4 \cdot 2^{g_{i}} \cdot 2^{f_{i}} - 4 \cdot 2^{g_{i}} \cdot 2^{f_{i+1}}}{4 \cdot 2^{f_{i+1}} - 2 \cdot 2^{f_{i}} - 2}.
   \end{align*}
      \begin{align*}
       \frac{x_{i+2}}{4 \cdot 2^{f_{i+3}} - 2 \cdot 2^{f_{i+2}} - 2} 
       &\leq \left(1 + \frac{1}{4 \cdot 2^{f_{i+1}} - 2 \cdot 2^{f_{i}} - 2}\right)x_{i} 
       + \frac{4 \cdot 2^{f_{i+2}} \cdot 2^{g_{i+2}} + 4 \cdot 2^{f_{i+3}} \cdot 2^{g_{i+2}}}{4 \cdot 2^{f_{i+3}} - 2 \cdot 2^{f_{i+2}} - 2} \\
       &\quad - \frac{4 \cdot 2^{f_{i}} \cdot 2^{g_{i}} + 4 \cdot 2^{f_{i+1}} \cdot 2^{g_{i}}}{4 \cdot 2^{f_{i+1}} - 2 \cdot 2^{f_{i}} - 2} \\
       &\leq \left(1 + \frac{1}{4 \cdot 2^{f_{i+1}} - 2 \cdot 2^{f_{i}} - 2}\right)x_{i} + 2^{g_{i+2}}.
   \end{align*}
   We conclude the following:
   \begin{align*}
       x_{i+2} 
       &\leq 4 \cdot 2^{f_{i+3}} \left(1 + \frac{1}{2 \cdot 2^{f_{i+1}}}\right)x_{i} + 4 \cdot 2^{g_{i+2}} \cdot 2^{f_{i+3}} \\
       &\leq 4 \cdot 2^{f_{i+3}}x_{i} + 2 \cdot 2^{f_{i+3} - f_{i+1}}x_{i} + 4 \cdot 2^{f_{i+3}} \cdot 2^{g_{i+2}} \\
       &\leq 6 \cdot 2^{f_{i+3}}x_{i} + 4 \cdot 2^{f_{i+3}} \cdot 2^{g_{i+2}} \\
       &\leq \sum_{k=0}^{\frac{i+2}{2}} 4 \cdot 2^{g_{2k} + \sum_{j=k}^{\frac{i+2}{2}} f_{2j+1}} \cdot 6^{\frac{i - 2k + 2}{2}}.
   \end{align*}
   The capture time $T$ would be as follows:
   \begin{align*}
       &\sum_{j=0}^{\frac{i-3}{2}} x_{2j+1} + 2 \cdot \left(2^{f_{i}} \left(2^{g_{i}} + v_{i} \cdot \sum_{j=0}^{\frac{i-3}{2}} x_{2j+1}\right)\right) 
       + \frac{d}{1-v} + \frac{\sum_{j=0}^{\frac{i-3}{2}} x_{2j+1} v}{1-v} \\
       &\quad + \frac{2 \cdot 2^{f_{i}} \cdot 2^{g_{i}}v}{1-v} 
       + \frac{2 \cdot 2^{f_{i}} \cdot v_{i} \cdot \sum_{j=0}^{\frac{i-3}{2}} x_{2j+1} v}{1-v} \\
       &\leq \sum_{j=0}^{\frac{i-3}{2}} x_{2j+1} + 2 \cdot 2^{f_{i}} \cdot 2^{g_{i}} 
       + 2 \cdot 2^{f_{i}} \cdot v_{i} \cdot \sum_{j=0}^{\frac{i-3}{2}} x_{2j+1} \\
       &\quad + \frac{d}{1-v} 
       + v \cdot \frac{\sum_{j=0}^{\frac{i-3}{2}} x_{2j+1} + 2 \cdot 2^{f_{i}} \cdot 2^{g_{i}} + 2 \cdot 2^{f_{i}} \cdot v_{i} \sum_{j=0}^{\frac{i-3}{2}} x_{2j+1}}{1-v}.   
   \end{align*}
   Since $v_{i} = 1 - 2^{-f_{i}}$, we get:
   \begin{align*}
    T &\leq \frac{2 \cdot 2^{f_{i}} \cdot \sum_{j=0}^{\frac{i-3}{2}} x_{2j+1} - \sum_{j=0}^{\frac{i-3}{2}} x_{2j+1} + d + 2\cdot 2^{f_{i}}2^{g_{i}}}{1-v}.
   \end{align*}
   We conclude that the competitive ratio is as follows:
   \begin{align*}
       CR &= \frac{2 \cdot 2^{f_{i}}2^{g_{i}}}{d} + 1 + \frac{2 \cdot 2^{f_{i}} \cdot \sum_{j=0}^{\frac{i-3}{2}} x_{2j+1}}{d} - \frac{\sum_{j=0}^{\frac{i-3}{2}} x_{2j+1}}{d} \\
       &\leq \frac{2 \cdot 2^{f_{i}}2^{g_{i}}}{d} + 1 + \frac{2 \cdot 2^{f_{i}} x_{i}}{2 \cdot 2^{f_{i+1}}} 
       \leq \frac{2 \cdot 2^{f_{i}}2^{g_{i}}}{d} + 1 + \frac{x_{i}}{2^{f_{i}}} \\
       &\leq \frac{2 \cdot 2^{f_{i}}2^{g_{i}}}{d} + 1 + \left(\frac{1}{2^{f_{i}}}\right)\left(\sum_{k=0}^{\frac{i}{2}} 4 \cdot 2^{g_{2k} + \sum_{j=k}^{\frac{i}{2}} f_{2j+1}} \cdot 6^{\frac{i - 2k}{2}}\right) \\
       &\leq 2\cdot 2^{f_{i}}2^{g_{i-1}} + 1 + 4\left(\frac{1}{2^{f_{i}}}\right)\left(\frac{i}{2} + 1\right) 2^{g_{i}} \cdot 2^{\sum_{j=0}^{\frac{i}{2}} f_{2j+1}} \cdot 6^{\frac{i}{2}}.
   \end{align*}
   Considering the fact that the target is found in iteration $i$, this gives $2^{2^{i-1}} \leq \frac{1}{1-v}$.
   Considering $g_{i} = f_{i} = 2^{i}$, we also have the following:
   \begin{align*}
       2^{g_{i}} &\leq \frac{1}{(1-v)^2}, \\
       \sum_{j=0}^{\frac{i}{2}} f_{2j+1} &= \sum_{j=0}^{\frac{i}{2}} 4^j \cdot 2 
       = 2 \cdot \frac{4^{\frac{i}{2}+1} - 1}{3} 
       = \frac{8 \cdot 4^{\frac{i}{2}}}{3} - \frac{2}{3} 
       = \frac{8 \cdot 2^{i}}{3} - \frac{2}{3}, \\
       6^{\frac{i}{2}} &\leq 6^{\log \log(\max(d,\frac{1}{1-v}))}
   \end{align*}
   We conclude that the competitive ratio becomes as follows:
   \begin{align*}
       CR &\leq 1 + 2\cdot 2^{3\cdot2^{i-1}} + 4\left(\frac{i}{2} + 1\right) \cdot 2^{g_{i}} \cdot 2^{\frac{5}{3} \cdot 2^{i}} \cdot 2^{-\frac{2}{3}} \cdot 6^{\frac{i}{2}} \\
       &\leq 1 + 2\cdot 2^{3\cdot2^{i-1}} + 2^{\frac{4}{3}}\left(\frac{i}{2} + 1\right) 2^{\frac{8}{3} \cdot 2^{i}} \cdot 6^{\frac{i}{2}} \\
       &\leq 1 + 2\cdot 2^{3\cdot2^{i-1}} + 2^{\frac{4}{3}}\left(\frac{i}{2} + 1\right) 2^{\frac{8}{3} \cdot 2^{i}} \cdot 6^{\frac{i}{2}}\\       
&\leq 1+2\cdot M^3+2^{\frac{4}{3}}\cdot \log \log^{\frac{3}{2}} M \cdot M^{\frac{16}{3}}\cdot 6^{\log\log M}  ,       
       \end{align*}
where $M = \max\{d,\frac{1}{1-v}\}$.
This final expression provides the upper bound for the competitive ratio, considering the dependence on the parameters of the problem. The result encapsulates the complexities introduced by the exponential growth of terms and the logarithmic dependencies on $v$. 
This completes the proof of the theorem.~\qed
\end{proof}
\section{Conclusion}
We considered the problem of capturing an oblivious moving target on an infinite line with two robots in the S/R model. Two cases were analyzed based on the target's movement: either moving towards or away from the origin. For each case, we took into account various constraints related to the knowledge the robots have about the target's speed and its initial distance from the origin. 
Establishing tight bounds for the scenario where the distance is unknown and the target is moving away from the origin remains an open problem. As a topic for future research it would be interesting to study competitive ratios in linear search for multi-robot systems and for capturing a mobile target with robots subjected to either crash or byzantine faults.

\bibliographystyle{abbrv}
\bibliography{refs}
\end{document}